\documentclass[envcountsame]{llncs}

\usepackage{amssymb}
\usepackage{graphicx}
\usepackage{amsmath}

\begin{document}

\mainmatter

\title{Contextuality and the Weak Axiom in the Theory of Choice}

\author{William Zeng\inst{1} \and Philipp Zahn\inst{2}}

\institute{Department of Computer Science, University of Oxford, Oxford, UK\\
\email{william.zeng@cs.ox.ac.uk},
\and
School of Economics and Political Science, University of St. Gallen, Switzerland\\
\email{philipp.zahn@unisg.ch}}
\maketitle

\vspace{-15pt}
\begin{abstract}
Recent work on the logical structure of non-locality has constructed scenarios where observations of multi-partite systems cannot be adequately described by compositions of non-signaling subsystems. In this paper we apply these frameworks to economics.  First we construct a empirical model of choice, where choices are understood as observable outcomes in a certain sense.  An analysis of contextuality within this framework allows us to characterize which scenarios allow for the possible construction of an adequate global choice rule.  In essence, we mathematically characterize when it makes sense to consider the choices of a group as composed of individual choices. We then map out the logical space of some relevant empirical principles, relating properties of these \emph{contextual choice scenarios} to no-signalling theories and to the weak axiom of revealed preference.
\end{abstract}

\vspace{-25pt}
\section{Introduction}
In this paper we uncover a connection between observed choices in economics and empirical models in quantum physics. In particular, we show the precise relationships between the weak axiom of revealed preferences, a consistency property often imposed in choice theory, and contextuality and no-signalling conditions on measurements in the foundations of quantum physics.

The framework we use is borrowed from work in the foundations of quantum mechanics where a general logical theory of contextuality has been developed in recent years. Our work applies the empirical models and the framework of contextuality beyond quantum physics or computer science. In this sense, we are furthering the work of Abramsky et. al. \cite{sheaftheoretic,logicalBell,abramskycontextual,cohom,mansfieldDPhil} by showing that observed economic choices can be seen as one instantiation of an abstract contextual semantics.  

Our paper is related to two strands of work. On the one hand there is a growing literature that considers the consequences of quantum-resources for economic theory. This interaction between quantum foundations and the foundations of economics, where agents can make use of quantum resources, has led to results in quantum games \cite{Brunner2012}, decision theory \cite{quantumDecisionTheory}, voting systems \cite{quantumVoting}, and other areas.

On the other hand there is a well developed literature that combines classical decision theory with elements from quantum mechanics to address various empirical puzzles in a unified theoretical framework~\cite{Busemeyer_Bruza2012}. Quantum theory is used, for instance, to model ambiguity and fundamental uncertainty. In~\cite{LambertMogiliansky2009}, for example, players' preferences are allowed to be indeterminate before they make a decision. Moreover, various authors have translated probabilistic choice models into a quantum setting to account for empirical violations of classical choice theory. For instance,~\cite{Pothos2009} provides a quantum probability framework that accounts for violations of the sure-thing principle in experiments. Order effects of measurements on behavior and attitudes also have been addressed, for example, in~\cite{Wang2013a}. The utility of quantum mechanics (and often its contextual nature) in addressing these specific issues, helps motivate our move to study contextuality in general, rather than just as a part of a quantum mechanical model. 

While we share the idea of using elements from quantum motivated settings for classical decision theory, the focus in our paper is different. We consider only the observed choices of agents from particular empirical scenarios. This means that we do not provide an internal model of the agent. This focus is similar to the perspective taken by \cite{Danilov2008}, but, in contrast to their focus on measurement induced effects, our main investigation is to specify under which conditions it is possible to construct a sensible measurement at all, via an adequate global (context-independent) choice rule from local choices. Note also that while the presentation of our framework focuses on individual agents, it can easily be extended to groups of agents via the connection between decision theory and voting theory \cite{saari2005profile}.

The paper is structured as follows:
We first describe the mathematical framework of contextuality using empirical models. Next, we define choice scenarios and show how observed choices fit into this setting. In Sections 4 and 5 we state the general definitions and theorems that characterize the relationship between contextual semantics and choice scenarios. We show that, while choice scenarios that do not obey the weak axiom can be either contextual or non-contextual, choice scenarios that do not obey the weak axiom must be contextual. We also show that the weak axiom is strictly weaker than the no-signalling condition. Moreover, we show that under sufficient overlap of budgets, made precise in the paper, the weak axiom and no-signalling are equivalent. Section 6 briefly indicates how our setting can be extended to include probabilistic choices and mentions how they can be characterized by the logical Bell inequalities of \cite{logicalBell}. We conclude with a discussion that interprets these results in economic terms.
\vspace{-14pt}
\section{Mathematical Framework of Contextuality}

In~\cite{logicalBell}, the authors generalize a notion of contextuality from the quantum mechanical setting into an abstractly logical one that can be applied to many empirical scenarios. 

\begin{definition}
An \emph{{\bf empirical scenario}} is given by
\begin{enumerate}
\item a set of measurements $X$
\item a set $O$ of possible outcomes for the measurements
\item subsets $U\subseteq{X}$ that represent possible \emph{measurement contexts}
\item $\mathcal{U}$, a subset of the powerset of $X$ that defines the set of all possible measurement contexts. We will call this the \emph{set of feasible experiments}.
\end{enumerate}
\end{definition}

\begin{example}
As a simple example, we consider an empirical scenario given by two systems $A$ and $B$.  We further posit that on each system we have a choice of two different measurements, each of which has outcomes either $0$ or $1$, i.e. $O=\{0,1\}$. These systems could, for example, be two coins and we could either check if a coin is heads or weigh it to check if it is heavier than 1 gram.  This means our empirical scenario consists of four boolean variables, $X = \{a,b,a',b'\}$ where:
\begin{description}
\item $a$ is 1 iff the first coin is heads.
\item $a'$ is 1 iff the first coin weighs over a gram.
\item $b$ is 1 iff the second coin is heads.
\item $b'$ is 1 iff the second coin weighs over a gram.
\end{description}
As our scenario only allows us to choose one of the two measurements on the coin at a time, an example of a measurement context is $U=\{a',b'\}$, corresponding to weighing both coins. The complete set of feasible experiments is given by these two element subsets of $X$:
\[ \mathcal{U} = \{\{a,b\},\{a,b'\},\{a',b\},\{a',b'\}\}. \]

Joint outcomes for experiments are then given by functions from measurement contexts to the set of outcomes.  In our example, an example function
\[ \{a\mapsto 1, b\mapsto1 \} \]
represents measuring both coins to be heads.
\end{example}

\begin{definition}
Given an empirical scenario where $O=\{0,1\}$, a \emph{{ \bf binary empirical model}} is a map $C:\mathcal{U}\to\mathcal{P}(\mathcal{P}(X))$ such that
\[ U\mapsto C(U), \]
where $C(U)$ is a set of subsets of measurement context $U$ that could have outcome 1.\footnote{In other literature, $C(U)$ is referred to as the \emph{support} of the measurement context $U$ under the model $C$.}
\end{definition}

\begin{example}
Consider our coin scenario, but where the first coin is a double-sided heads coin. We can specify the binary empirical model $C$ explicitly:
\begin{align*}
\{a,b\}   \mapsto \{a,\{a,b\}\} \\
\{a',b\}  \mapsto \{\phi, a', b, \{a',b\}\} \\
\{a,b'\}  \mapsto \{a,\{a,b'\}\} \\
\{a',b'\} \mapsto \{\phi, a', b', \{a',b'\}\}
\end{align*}
We can more easily represent the empirical model with a table like the following:

\begin{tabular}{p{2cm}|p{2cm}|p{2cm}|p{2cm}|p{2cm}}
         & (0,0) & (1,0) & (0,1) & (1,1) \\\hline
 (a,b)   & 0     &  1    & 0     & 1  \\
 (a',b)  & 1     &  1    & 1     & 1  \\
 (a,b')  & 0     &  1    & 0     & 1  \\
 (a',b') & 1     &  1    & 1     & 1  \\
 \end{tabular}
 \newline\newline
In this table, the choice of row denotes a measurement context.  Each column then represents a particular outcome.  For example, the left-most column corresponds to an outcome of 0 for both systems.  The column to its right corresponds to an outcome of 1 for the first system and an outcome of 0 for the second. This would mean that $\{a\} \in C(U)$. The value of each cell in the table is defined by the following rule: 1 if the set of measurements with outcome 1 is contained in $C(U)$ and zero otherwise.  In our example the double-headed first coin means that it is impossible for us to get a 0 outcome for the first coin when we check its side.  This gives the zeros on the first and third columns.
\end{example}

\subsection{Contextuality}
It is important to note that we do not have completely free choice in choosing the empirical model as some empirical models lead to contradictions.  We  reproduce one such example from \cite{logicalBell} that is based on the Hardy paradox about possible outcomes for certain quantum mechanical systems.\footnote{The Greenberger-Horne-Zeilinger states \cite{GHZ}.}

\begin{example}
In this example we again have two systems, measurements $X = \{a,b,a',b'\}$, and outcomes $O=\{0,1\}$.  Consider the binary empirical model given by the following table:
\newline
\begin{tabular}{p{2cm}|p{2cm}|p{2cm}|p{2cm}|p{2cm}}
         & (0,0) & (1,0) & (0,1) & (1,1) \\\hline
 (a,b)   & 1     &  1    & 1     & 1  \\
 (a',b)  & 0     &  1    & 1     & 1  \\
 (a,b')  & 0     &  1    & 1     & 1  \\
 (a',b') & 1     &  1    & 1     & 0  \\
 \end{tabular}\label{tab:exp_doublecoin}
 \newline\newline
 The model specified by this table is logically inconsistent by the following reasoning.  Interpret outcome 0 as false and outcome 1 as true and consider the following formulas:
 \[ a\wedge b, \qquad \neg(a\wedge b'), \qquad\neg(a'\wedge b), \qquad a'\vee b'.\] 
According to the empirical model these should all be possible (true).  However, it is impossible to find individual assignments of $a,b,a',b'$ to be true or false that manifest this \cite{logicalBell}. This observation forces us to conclude that the given empirical model cannot be constructed by a composition of systems with defined values for $a,b,a',$ and $b'$.  This empirical model is \emph{contextual}.
\end{example}
 
In short, should we encounter a system that fits a contextual model, then we know that its behavior is not modelled by the composition of a series of separate subsystems.  Likewise, we know that a composition of individual systems will never generate contextual behavior.

\section{Choice}
This section applies the above empirical framework to an economic choice setting over a set of alternatives $X$. An agent has to choose from menus of alternatives which comprise subsets of all possible alternatives. 

\begin{definition}
A \emph{{\bf choice scenario}} is defined as the following:
\begin{enumerate}
\item A set of alternatives $X$.
\item A set $O = \{0,1\}$. $1$ indicates that an element $x\in X$ is possibly chosen.
\item Subsets $U\subseteq{X}$ that represent possible \emph{menus} an agent can choose from.
\item A set $\mathcal{U}$ of menus that represents the \emph{set of feasible menus} an agent can face.
\item A global choice rule $C:\mathcal{U}\to\mathcal{P}(\mathcal{P}(X))$ such that $U\mapsto C(U)$. This means that $C(U)$ is the set of subsets of elements in a given menu $U$ that are possibly chosen, i.e. whose outcome is $1$.
\end{enumerate}

\end{definition}

The set of alternatives $X$ can represent various types of choice problems. It could be a set of consumer goods, a list of political candidates, a set of survey questions, etc. The crucial point here is that it may be impossible to get a complete answer to the whole set $X$ given only answers to particular menus one at a time. There may be various reasons for this impossibility and we will illustrate with some examples.

When only partial information can be gathered, the key question is: Can we infer from the local choices what an agent would choose if he or she faced the complete set of alternatives? Under which conditions is it impossible to aggregate local choices in a coherent way?

\begin{example}\label{exp:wine}
Consider the following setting. A retailer offers wine at two periods in time. Availability of wines or other concerns may naturally dictate which menus he can offer at a given time. Let the alternatives be given by
\[X=\left\{Riesling14, Pinot Blanc14, Riesling13, Pinot Blanc13\right\}.\]   We will use the shorthand labels $a,a',b,b'$ for the wines respectively so that

\begin{description}
\item $a$ is 1 iff Riesling13 is chosen.
\item $a'$ is 1 iff Pinot Blanc13 is chosen.
\item $b$ is 1 iff Riesling14 is chosen.
\item $b'$ is 1 iff Pinot Blanc14 is chosen.
\end{description}

In such a choice scenario, choices could be observed to follow a contextual rule as in Example \ref{tab:exp_doublecoin}. The contradiction in this scenario is precisely the same one in the physical measurement setting. 
\end{example}

Another example would be to consider two waiters who are taking orders, one takes orders for food, and the other takes orders for beverages such that $X=\left\{beef,cake,wine,coffee\right\}$.

What is the economic meaning of a global choice rule? Its simplest interpretation is that if it were possible to let an agent choose from the global menu, where all alternatives are available at the same time, the global choice rule determines what the agent will choose. In the following section we make precise when one can aggregate local choices into global ones. 

Note that the impossibility of aggregating choices coherently depends on the analyst's definition of the set of alternatives. In the example above, a redefinition of the goods would have cleared the contradictions. However, this presupposes an understanding of the situation and background knowledge which, though obvious in these examples, is not necessarily available to the analyst in general. Consider another example where a coherent aggregation is obviously impossible. 

\begin{example}\label{exp:Luce}[due to Luce and Raiffa \cite{luce1957games}]
Suppose an agent is going to a restaurant twice. The set of meals is ${Salmon, Steak, Frog Legs, Fried Snails}$. The first time he faces the following menu ${Salmon, Steak}$ and chooses ${Salmon}$. But the next time, facing the menu ${Salmon, Steak, Frog Legs}$, he chooses ${Steak}$. 
\end{example}

In contrast to the examples before, the reasons why a global choice rule does not exist is less obvious. For instance, it could mean that crucial information is missing which would explain the agent's choices, or it could be the case that the agent is choosing in an incoherent way: potentially his choices are affected by the context itself.\footnote{There are many alternative explanations. See also \cite{Danilov2008} on the possibility that the measurement itself is changing the agent's preferences.} All of these choice examples and any others can be brought under the unified heading of empirical scenarios. 
\begin{theorem}
\label{thm:connection}
Every choice scenario defines a unique empirical scenario with binary empirical model and vice versa.
\end{theorem}
\begin{proof}
This is clear by inspecting the definitions.  The following table provides a glossary of sorts:
\begin{align*}
\begin{tabular}{p{5cm}|p{5cm}}
 Choices & Measurements  \\ \hline
 alternatives & measurements \\
 choice or non-choice  & outcomes \\
 menus  & measurement contexts \\
 all feasible menus & all feasible experiments \\
 global choice rule & binary measurement model \\
 \end{tabular}\label{glossary}
\end{align*}
\end{proof}
\section{Generalized Choice Contextuality}
Having illustrated by example the connection between empirical scenarios and choice settings, we can leverage this correspondence more generally. We begin with some definitions of contextuality from the quantum foundations literature\cite{logicalBell}.

A restriction of a function $s:X\to \{0,1\}$ to $U\subseteq X$ will be written $s|X$.  The support of $s|X$ are all the elements that are mapped to $1$ by it, as in the following example.  Take $X = \{a,b,c,d\}$ and $s = \{a\mapsto 1,b\mapsto 1, c\mapsto 0, d\mapsto 0\}$. Consider the restriction to $U = \{a,b,d\}$. We obtain $s|U = \{a\mapsto 1,b\mapsto 1, d\mapsto 0\}$ and the support of $s|U$ is $\{a,b\}$.

\begin{definition}
Given an empirical scenario $(X,O,\mathcal{U})$ and a binary empirical model $C$, a \emph{{\bf global section}} is an assignment $s:X\to O$ such that for all $U\in\mathcal{U}$ the support of $s|U$ is in $C(U)$.
\end{definition}

In physical terms, this says that a global section\footnote{As noted in \cite{logicalBell}, the terminology global section arises because these binary empirical models can be given the structure of a presheaf.  More details on this approach can be found in \cite{sheaftheoretic}.} gives a specific outcome to every measurement that can be used to represent any particular measurement context.  In other words, the binary empirical model can be reproduced from restrictions of the global section. 

In the choice setting, a global section means that we are able to assign a choice or non-choice to each alternative in such a way as to reproduce the choices that are made when restricted to any particular menu.  In this sense, the existence of such a global section determines whether or not our choices depend on the menu with which they are presented.  The following makes this intuition more precise:

\begin{definition}[adapted from \cite{sheaftheoretic}]
\label{def:noncontext}
A \emph{binary empirical model} is \emph{{\bf possibilistically noncontextual}} if for every element $\eta \in C(U)$ for some $U$, there is a global section $s'$ such that $\eta$ is in the support of $s'|U$.
\end{definition}
When Definition \ref{def:noncontext} does not hold, the model is \emph{contextual}.

The correspondence between empirical models and choice scenarios from Theorem \ref{thm:connection} motivates the construction of the following definition: 

\begin{definition}[Contextual choice scenarios]
A \emph{choice scenario} is \emph{{\bf non-contextual}} if and only if for every element $\eta\in C(U)$ for some $U$, there is a global section $s'$ such that $\eta$ is in the support of $s'|U$. Otherwise the choice scenario is \emph{{\bf contextual}}.
\end{definition}
As a way of interpreting contextuality in economic terms, we can imagine non-contextual choice scenarios as those where we would make the same choices if we were presented every alternative at once in a single menu.

\begin{definition}[Strongly contextual choice scenarios]
A choice scenario is \emph{\bf {strongly contextual}} if there exists no global section.
\end{definition}
Scenarios that are strongly contextual\footnote{The notion of strong contextuality in contextual semantics comes from \cite{sheaftheoretic}. } do not even have some sub-part that can be modeled as choices made independent of context.

\section{Choice Scenarios and the Theory of Choice}
\label{sec:space}

So far, we have interpreted the measurement of choices as a purely empirical approach without any reference to a specific economic theory. In this section, we will link this empirical approach to the theory of choice. To this end, we need the notion of \emph{budgets}. Economic agents have wealth constraints such that not all alternatives may be affordable given a certain wealth level. Changing income, or changing prices of goods, may alter the set of alternatives that are available. In the following, we will give menus an alternative interpretation: they represent the budget; thus they represent the alternatives that are affordable for a given agent.\footnote{Obviously, one could also consider a blend of the two views: (i) what is presented in a menu and (ii) what is affordable? We focus on the extreme case for simplicity.}   

When menus and budgets coincide, choice scenarios will always allow us to capture the observed behavior of agents acting according to the usual economic choice structures \cite[p.9]{mas-colell_microeconomic_1995}.  In fact, choice scenarios are more general than the usual choice structures, as the following is clear by definition:
\begin{proposition}
A choice structure is a choice scenario where $|C(U)|=1$ for all $U$.
\end{proposition}

\subsection{The Weak Axiom and Contextuality}

A central question in the theory of choice is how behavior changes over different budgets. As a rationality requirement, consistency of choices is imposed via the weak axiom of revealed preferences. In our choice scenario setting, this has the following form:

\begin{definition}
A choice scenario \emph{obeys the weak axiom} if for every pair of budgets $A,B$, with elements $x,y\in A\cap B$ such that $x\in C(A)$ and $y\in C(B)$ then
\[ x\in C(B). \]
\end{definition}

For a simple example where the weak axiom is violated reconsider Example \ref{exp:Luce}. As Salmon is preferred over Steak in the first menu, it should also be chosen in the larger menu ${Salmon, Steak, Frog Legs}$. 

In the following, we investigate the relationship between contextual choice scenarios and ones whose choice rules obey the weak axiom (Figure \ref{fig:venn}).

\begin{theorem}
\label{thm:context-wa}
Choice scenarios that do not obey the weak axiom are contextual if the set of budgets is closed under intersection.
\end{theorem}
\begin{proof}
As our scenario does not obey the weak axiom, we know that there exists a pair of budgets $A,B$, with elements $x,y\in A\cap B$ such that $x\in C(A)$ and $y\in C(B)$, but $x$ is not in $C(B)$. We will use this to demonstrate that there exists no global section that can adequately assign a value to $x$.

Suppose there were such a section $s:X\to O$. As $x\in C(A)$ but not in $C(B)$, we would require $x$ to be in the support of $s|A$ but not in the support of $s|B$. As budgets are closed under intersection, we reach a contradiction when trying to assign $x$ to the support of $s|A\cap B$. As we cannot construct a global section that satisfies the definition of possibilistic noncontextuality, the choice scenario must be contextual.
\end{proof}

\begin{theorem}
Choice scenarios that obey the weak axiom can be either contextual or non-contextual.
\end{theorem}
\begin{proof}
We show this with two examples. In the first, let budget $A=\{a,b\}$ and budget $B = \{a,c\}$. Choose a global section $s:\{a,b,c\}\to\{0,1\}$ that sends only $a$ to 1. A choice rule that has $C(A)=C(B)=\{a\}$ obeys the weak axiom and is non-contextual.
As the second example, let $A=\{a,b\}$ and $B = \{b,c\}$. Again, choose a global section $s:\{a,b,c\}\to\{0,1\}$ that sends only $a$ to 1. A choice rule that has $C(A)=\{a\}$ and $C(B)=\{b\}$ obeys the weak axiom and is contextual.
\end{proof}
In some sense, this decoupling of the weak axiom and contextuality results from the fact that the weak axiom only conditions the behavior of the choice rule for elements that are in $C(U)$, i.e. elements that are chosen.  Contextuality, on the other hand, requires a consistency over contexts where elements are not chosen as well, i.e. if an element is not chosen in some context, then it must also be not chosen in other contexts.

\subsection{The Weak Axiom and No-Signalling}

The weak axiom can also be related to the no-signalling condition for empirical models, whose definition from \cite{sheaftheoretic} we adapt to our setting. For a budget $A\in\mathcal{U}$, let $f_A:A\to\{0,1\}$ be the choice function that sends an element to 1 when it is chosen - i.e. is in some element of $C(A)$ - and sends an element to 0 when it is not.

\begin{definition}
A choice scenario $(X,\mathcal{U},C)$ is \emph{non-signalling} if and only if for any two budgets $A,B\in\mathcal{U}$:
\[f_A |( A\cap B) = f_B | (A\cap B),\]
i.e. all the choice functions have to agree when restricted to their intersections.
\end{definition}

In our framework, no-signalling choices are those where the choice or non-choice of an alternative must be consistent across all the budgets in which that alternative appears.  It is perhaps not surprising then that such a strong global consistency condition obeys the weak axiom, as is shown in the following theorem:
\begin{theorem}
\label{thm:nsImplies}
Non-signalling choice scenarios obey the weak axiom.
\end{theorem}
\begin{proof}
Using this notion of choice functions for budgets, the weak axiom states that if $f_A(x) = 1$ and $f_B(y) = 1$ then $f_B(x) = 1$. By the symmetry of the definition, it is easy to see that $f_A(y) = 1$ necessarily as well. If $f_A(x) = 1$ and $f_B(y) = 1$, then no-signalling implies the same.  
\end{proof}

Indeed, the no-signalling requirement is actually stronger then that required for the weak axiom (Figure \ref{fig:venn}).  There do exist choice scenarios where agents act rationally, but not exactly consistently over all budgets.
\begin{theorem}
The weak axiom is strictly weaker than no-signalling.
\end{theorem}
\begin{proof}
Consider the following example for $x,y\in A\cap B$. Let $f_A(x) = 1$ and $f_A(y)=f_B(x)=f_B(y)=0$.
\end{proof}
As in the previous discussion of the weak axiom and contextuality, the economic condition is weaker than the physically influenced one because non-choices are not required to be consistent.  One can freely not choose an alternative in one budget, but then choose it in another if there is no preferable option.  In an empirical scenario though, the observed outcome of "not-chosen" needs to be considered in consistency conditions.

\begin{figure}[t]
\centering
\includegraphics[width=0.5\textwidth]{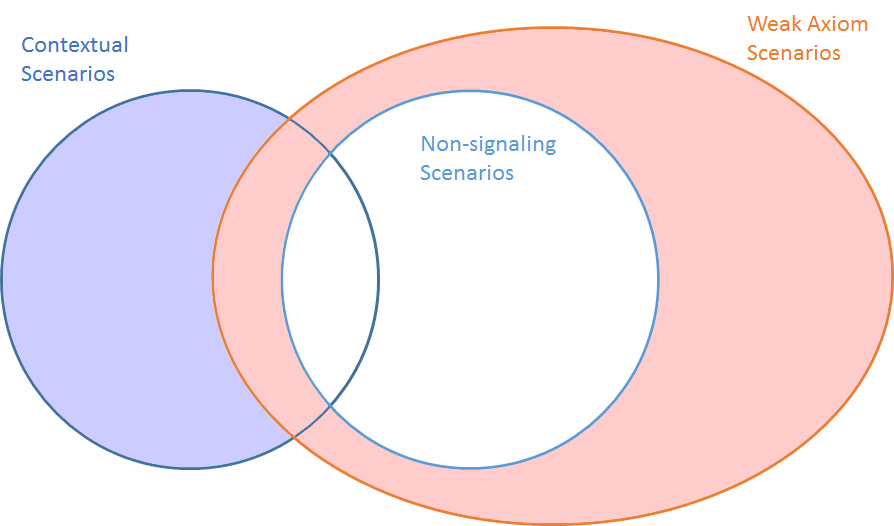}
\caption{The overlapping space of scenarios, summarizing the results of Section \ref{sec:space}.  Scenarios in the red region are specified in Theorem \ref{thm:wa-ns}. Scenarios in the blue region are charecterized by Theorem \ref{thm:context-wa}.}
\label{fig:venn}
\end{figure}

Still, we can give a clear picture of when the two are equivalent.
\begin{theorem}
\label{thm:wa-ns}
The weak axiom is equivalent to no-signalling if a choice scenario has the following property for all budgets $A$ and $B$:
\begin{align}
\label{prop:WANS}
\mbox{there exists not necc. distinct $x,y\in A\cap B$ s.t. $x\in C(A)$ and $y\in C(B)$}.
\end{align}
\end{theorem}
\begin{proof}
We show this by contradiction.  Consider a choice scenario that obeys the weak axiom and property (\ref{prop:WANS}) but that is not no-signalling.  A signalling choice scenario has some $z\in A\cap B$ such that $f_A(z)\neq f_B(z)$. Without loss of generality we can take $f_A(z)=1$ and $f_B(z) = 0$. By property (\ref{prop:WANS}), we can always find some $y\in A\cap B$ such that $f_B(y)=1$, yet the weak axiom states that if $f_A(z)=1$ and $f_B(y)=1$ then $f_B(z) = 1$ causing a contradiction.
\end{proof}
In economic terms this property (\ref{prop:WANS}) states that an agent, when presented with two budgets that have overlapping alternatives, must, in each budget, choose at least one alternative from the overlapping ones. 

Results from this section are summarized in Figure \ref{fig:venn}, showing how contextual, non-signalling, and weak axiom scenarios are interrelated. 

\section{Probabilistic Choices}

Our setting naturally generalizes to probabilistic choice scenarios.  Here the outcome set is $O = [0,1]$, where $1$ indicates that an element $x\in X$ is always chosen, and the choice rule $C$ is replaced with  probability distributions $d_U:X\to O$ for each context.  

Every probabilistic choice scenario can be reduced to an underlying choice scenario by taking $C(U)$ as the support of $d_U$. In this way we are able to label probabilistic choice scenarios contextual or strongly contextual according to the property of this underlying choice scenario.

In these probabilistic settings, we can use a single \emph{logical bell inequality} to test for contextuality \cite{logicalBell}. Further, these inequalities provide a metric for understanding the degree to which a scenario is contextual from how close a scenario is to the bound.  Scenarios that maximally violate a bell inequality are exactly the strongly contextual ones \cite{logicalBell}.
\section{Discussion}
In this paper, we have constructed a framework for analyzing the contextuality of choices using tools from quantum foundations, mapping out the logical space of some important physical and economic principles.  We show that a series of logical Bell inequalities can be applied to understand the degree of contextuality of choice scenarios.  In economic terms, these results are of particular interest for the following reasons:

The observation of contextual choices can imply the non-existence of consistent internal models of choice.  To the extent that we believe individual agents not to be purely motivated by individual preferences, measures of contextuality can act as measurements of collusion as they categorically represent behaviors that cannot be generated from fixed individual choices.

To the extent that we doubt agents are well modelled by individual preferences, contextual frameworks provide a setting to investigate more general alternatives without sacrificing rigor and specificity. The foundations of quantum computation, that are well suited to deal with contextualities,  provide a rich toolbox for economics to address these problems.

Even if the phenomena of contextuality is, at present, rarely observed it is still relevant because we know that it is possible to physically implement it with quantum systems. One could program quantum computers to make contextual choices and, by so doing, allow for more behaviors than can be modelled classically if it proves advantageous.

\subsection*{Acknowledgements}
The authors thank Samson Abramsky for suggesting investigation of the no-signalling condition. Zeng gratefully acknowledges the support of The Rhodes Trust and Zahn gratefully acknowledges financial support by the Deutsche Forschungsgemeinschaft (DFG) through SFB 884 ``Political Economy of Reforms''.

\bibliographystyle{splncs03}
\bibliography{bibliography}

\begin{thebibliography}{10}
\providecommand{\url}[1]{\texttt{#1}}
\providecommand{\urlprefix}{URL }

\bibitem{sheaftheoretic}
{Abramsky}, S., {Brandenburger}, A.: {The sheaf-theoretic structure of
  non-locality and contextuality}. New Journal of Physics  13(11), 113036 (Nov
  2011)

\bibitem{logicalBell}
{Abramsky}, S., {Hardy}, L.: {Logical Bell inequalities}. Phys. Rev. A  85(6),
  062114 (Jun 2012)

\bibitem{abramskycontextual}
Abramsky, S.: Contextual semantics: From quantum mechanics to logic, databases,
  constraints, and complexity. Bulletin of EATCS  2(113) (2014)

\bibitem{cohom}
Abramsky, S., Mansfield, S., Soares~Barbosa, R.: The cohomology of non-locality
  and contextuality. Proceedings of Quantum Physics and Logic 2011 pp. 1--14
  (2011), electronic Proceedings in Theoretical Computer Science (EPTCS)

\bibitem{quantumVoting}
Bao, N., Halpern, N.Y.: Quantum voting and violation of arrow's impossibility
  theorem. arXiv preprint arXiv:1501.00458  (2015)

\bibitem{quantumDecisionTheory}
Brandenburger, A., La~Mura, P.: Quantum decision theory. arXiv preprint
  arXiv:1107.0237  (2011)

\bibitem{Brunner2012}
Brunner, N., Linden, N.: Connection between bell nonlocality and bayesian game
  theory. Nature communications  4 (2013)

\bibitem{Busemeyer_Bruza2012}
Busemeyer, J.R., Bruza, P.D.: Quantum Models of Cognition and Decision.
  Cambridge University Press (2012)

\bibitem{Danilov2008}
Danilov, V., Lambert-Mogiliansky, A.: {Measurable systems and behavioral
  sciences}. Mathematical Social Sciences  55(3),  315--340 (May 2008)

\bibitem{GHZ}
Greenberger, D.M., Horne, M.A., Zeilinger, A.: Going beyond bell’s theorem.
  In: Bell’s theorem, quantum theory and conceptions of the universe, pp.
  69--72. Springer (1989)

\bibitem{luce1957games}
Luce, R., Raiffa, H.: Games and Decisions: Introduction and Critical Survey.
  Dover books on advanced mathematics, Dover Publications (1957)

\bibitem{mansfieldDPhil}
Mansfield, S.: The mathematical structure of non-locality and contextuality.
  D.Phil. thesis, Oxford University (2013)

\bibitem{mas-colell_microeconomic_1995}
Mas-Colell, A., Whinston, M.D., Green, J.: Microeconomic Theory. Oxford
  University Press (1995)

\bibitem{LambertMogiliansky2009}
Mogiliansky, A.L., Zamir, S., Zwirn, H.: Type indeterminacy: A model of the
  kt(kahneman–tversky)-man. Journal of Mathematical Psychology  53(5),  349
  -- 361 (2009), special Issue: Quantum Cognition

\bibitem{Pothos2009}
Pothos, E.M., Busemeyer, J.R.: {A quantum probability explanation for
  violations of 'rational' decision theory.} Proceedings. Biological sciences /
  The Royal Society  276(1665),  2171--2178 (2009)

\bibitem{saari2005profile}
Saari, D.G.: The profile structure for luce's choice axiom. Journal of
  Mathematical Psychology  49(3),  226--253 (2005)

\bibitem{Wang2013a}
Wang, Z., Busemeyer, J.R.: {A quantum question order model supported by
  empirical tests of an a priori and precise prediction}. Topics in Cognitive
  Science  5(4),  689--710 (2013)

\end{thebibliography}

\end{document}